\documentclass[a4paper,11pt]{scrartcl}

\usepackage{etoolbox}
\newtoggle{ARXIV}
\toggletrue{ARXIV}

\iftoggle{ARXIV}{
\usepackage[utf8]{inputenc}
\usepackage{amsthm}
\theoremstyle{plain}
\newtheorem{theorem}{Theorem}
\newtheorem{definition}[theorem]{Definition}
\newtheorem{lemma}[theorem]{Lemma}
\newtheorem{example}[theorem]{Example}
\newtheorem{corollary}[theorem]{Corollary}
}{}

\usepackage{amsmath,amssymb}
\usepackage{ifthen}

\newcommand{\pn}[1]{\textsc{#1}}

\newcommand{\problem}[3]{
\vspace{\topsep}
\noindent\fbox{\begin{minipage}{.8\textwidth}
\pn{#1}\\
 \textbf{Input:} #2\\
\textbf{Problem:} #3
\end{minipage}
}
\vspace{\topsep}
}

\newcommand{\myindex}[2][]{%
\ifthenelse{\equal{#1}{}}{\index{#2}#2}{\index{#2}#1}}

\newcommand{\eindex}[1]{\emph{\myindex{#1}}}

\newcommand{\rel}[1]{\mathbf{#1}}
\newcommand{\relA}{\rel{A}}
\newcommand{\relB}{\rel{B}}
\newcommand{\relD}{\rel{D}}
\newcommand{\calC}{\mathcal{C}}
\newcommand{\calG}{\mathcal{G}}
\newcommand{\calH}{\mathcal{H}}
\newcommand{\calN}{\mathcal{N}}
\newcommand{\calQ}{\mathcal{Q}}

\newcommand{\id}{\mathrm{id}}

\newcommand{\var}[1]{\textsf{var}({#1})}
\newcommand{\dom}[1]{\textsf{dom}({#1})}

\newcommand{\aug}[1]{\textsf{aug}({#1})}

\newcommand{\smallp}{\ensuremath{p\textup{-}}}

\newcommand{\CQ}{\mathrm{CQ}}
\newcommand{\sCQ}{\mathrm{\#CQ}}

\newcommand{\sCQh}{\mathrm{\#CQ}_{hyp}}

\newcommand{\clique}{\pn{Clique}}
\newcommand{\sclique}{\pn{\#Clique}}
\newcommand{\pclique}{\smallp\clique}
\newcommand{\psclique}{\smallp\sclique}

\newcommand{\contract}{\mathsf{contract}}

\theoremstyle{plain}
\newtheorem{observation}[theorem]{Observation}
\newtheorem{proposition}[theorem]{Proposition}
\newtheorem{claim}{Claim}

\newcommand{\hatrelA}{\hat{\relA}}
 \newcommand{\hatrelB}{\hat{\relB}}
 \newcommand{\hatcalC}{\hat{\calC}}
 \newcommand{\hatR}{\hat{R}}
 \newcommand{\hatA}{\hat{A}}
 \newcommand{\hatB}{\hat{B}}
 
\newcommand{\FPT}{\mathsf{FPT}}
\newcommand{\sW}[1]{\mathsf{\#W[#1]}}
\newcommand{\W}[1]{\mathsf{W[#1]}}
\newcommand{\NP}{\mathsf{NP}}
\newcommand{\p}{\mathsf{FP}}

\newcommand{\comment}[2]{}
\newcommand{\hubie}[1]{\comment{\color{red}Hubie}{\color{red}#1}}
\newcommand{\stefan}[1]{\comment{\color{red}Stefan}{\color{red}#1}}

\newcommand{\N}{\mathbb{N}}
\renewcommand{\dom}{\mathsf{dom}}
\newcommand{\powfin}{\wp_{\mathsf{fin}}}
\newcommand{\lpr}{\langle}
\newcommand{\rpr}{\rangle}
\newcommand{\prom}[1]{\mathsf{prom}\textup{-}#1}
\newcommand{\param}[1]{\mathsf{param}\textup{-}#1}
\newcommand{\paramprom}[1]{\param{\prom{#1}}}

\iftoggle{ARXIV}{}{
\volumeinfo{Marcelo Arenas and Martin Ugarte}
{2}
{18th International Conference on Database Theory (ICDT'15)}
{31}
{1}
{1}
\EventShortName{ICDT'15}
\DOI{10.4230/LIPIcs.ICDT.2015.1}
}

\title{A Trichotomy in the Complexity of Counting Answers to Conjunctive Queries}
\iftoggle{ARXIV}
{
\author{
Hubie Chen \footnote{Universidad del País Vasco, E-20018 San Sebastián, Spain,
\emph{and} IKERBASQUE, Basque Foundation for Science, E-48011 Bilbao, Spain}
 \and
Stefan Mengel\footnote{LIX UMR 7161,
\'{E}cole Polytechnique, Universit\'e Paris Saclay,
France, partially supported by a Qualcomm grant administered by \'Ecole Polytechnique. }} 

}
{
 \author[1]{Hubie Chen}
\author[2]{Stefan Mengel\footnote{Partially supported by a Qualcomm grant administered by \'Ecole Polytechnique. }}
\authorrunning{Hubie Chen and Stefan Mengel}

\affil[1]{
Universidad del País Vasco,
E-20018 San Sebastián,
Spain,
\emph{and} IKERBASQUE, Basque Foundation for Science,
E-48011 Bilbao,
Spain
}
\affil[2]{LIX UMR 7161,
\'{E}cole Polytechnique,
France}


\Copyright{Hubie Chen and Stefan Mengel}
}

\begin{document}

\maketitle

\begin{abstract}
Conjunctive queries are basic and heavily studied database queries;
in relational algebra, they are the select-project-join queries.
In this article, we study the fundamental problem of counting,
given a conjunctive query and a relational database,
the number of answers to the query on the database.
In particular, we study the complexity of this problem relative
to sets of conjunctive queries.
We present a trichotomy theorem, which shows essentially that
this problem on a set of conjunctive queries is either
tractable, equivalent to the parameterized CLIQUE problem, or
as hard as the parameterized counting CLIQUE problem;
the criteria describing which of these situations occurs
is simply stated, in terms of graph-theoretic conditions.
\end{abstract}

\section{Introduction}

Conjunctive queries are the most basic and most heavily studied database
queries.
They can be formalized logically as formulas consisting of a
sequence of existentially quantified variables, followed by a
conjunction of atomic formulas; in relational algebra, they are
the \emph{select-project-join} queries (see e.g.~\cite{AbiteboulHullVianu95-foundationsdatabases}).
Ever since the landmark 1977 article of 
Chandra and Merlin~\cite{ChandraM1977},
complexity-theoretic aspects of conjunctive queries have been a
research subject of persistent and enduring interest which
continues to the present day (as a sampling, we point to the works~\cite{koalititsvardicontainment,PapadimitriouYannakakis99-database,GottlobLeoneScarcello02-hypertreedecomposisions,Grohe07,SchweikardtSchwentickSegoufin09-querylanguages,Marx13-tractablehypergraph,PichlerS13,ChenMueller14};
see the discussions and references therein for more information).
The problem of evaluating a Boolean (closed) conjunctive query 
on a relational database
is equivalent to a number of well-known problems, including 
conjunctive query containment, the homomorphism problem
on relational structures, and the 
constraint satisfaction problem~\cite{ChandraM1977,koalititsvardicontainment}.
That this evaluation problem appears in many equivalent guises
attests to the foundational and primal nature of this problem, and
it has correspondingly been approached and studied
from a wide variety of perspectives and motivations.

In this article, we study the fundamental problem of 
counting, given a conjunctive query and a relational database,
the number of query answers, that is, the number of
assignments that make the query true with respect to the database;
we denote this problem by $\sCQ$. 
In addition to being a natural problem in its own right,
let us remark that
all practical query languages supported by database management systems
have a counting operator.
We study the complexity of $\sCQ$
relative to sets of conjunctive queries, that is,
we study a problem family: each set of conjunctive queries
gives rise to a restricted version of the general problem.
Our objective is to determine on which sets of conjunctive queries
$\sCQ$ is tractable, and more broadly,
to understand the complexity behavior of the
problem family at hand.
Throughout, we assume that each considered set of conjunctive queries
is of \emph{bounded arity}, by which we mean that
there is a constant bounding the arity of all relation symbols in 
all queries in the set.\footnote{  
It is known that when no such bound on the
arity
is assumed, the complexity of query evaluation can be highly sensitive
to the representation of database relations~\cite{chengrohe}.
In contrast, natural representations are equivalent under
polynomial-time
translation;
thus, the study of bounded arity queries can be viewed as
the investigation of the representation-independent case.
}

Surprisingly, despite the natural and basic character of the 
counting problem $\sCQ$, the project of understanding
its complexity behavior over varying sets of queries 
has been carried out in previous work
for only two particular types of 
conjunctive queries.
In the case of Boolean conjunctive queries, 
the problem $\sCQ$ specializes to the problem of deciding
whether or not such a query evaluates to 
\emph{true} or \emph{false} on a database.
A classification of sets of Boolean conjunctive queries
was given by Grohe~\cite{Grohe07},
showing essentially that this decision problem is
either polynomial-time tractable, or is hard under a typical
complexity-theoretic assumption from parameterized complexity,
namely, that the parameterized $\clique$ problem 
is not fixed-parameter
tractable (see Theorem~\ref{thm:grohe2}).\footnote{
This is equivalent to the assumption that the parameterized
complexity class W[1] is not contained in the parameterized
complexity class FPT, which is the phrasing that Grohe employs.
}
It is well-known that a conjunctive query can be naturally
mapped to a relational structure 
(see Definition~\ref{def:natural-model}); 
the tractable sets of Grohe's classification are such queries 
whose corresponding structures have 
\emph{cores} of \emph{bounded
  treewidth}.
The \emph{core} of a structure $\relA$ is, intuitively, the smallest structure that
is (in a certain sense) equivalent to the structure $\relA$,
and \emph{bounded treewidth} is a graph-theoretical condition 
that can intuitively be taken as a notion of tree similitude.
Following Grohe's work, Dalmau and Jonsson~\cite{DalmauJ04}
studied the case of quantifier-free conjunctive queries
(which they phrase as the problem of counting homomorphisms
between a given pair of relational structures).
They proved that bounded treewidth is the property that
determines polynomial-time tractability for this case; 
in contrast to Grohe's theorem,
the statement of their classification 
(Theorem~\ref{thm:dalmauJ})
does not refer to the notion of
core, and is proved under the assumption that the 
\emph{counting} version 
$\sclique$
of the parameterized $\clique$ problem
is not fixed-parameter tractable.

In this article, we present a trichotomy theorem describing the
complexity of the counting problem $\sCQ$ on 
each possible set of conjunctive queries.
Our trichotomy theorem unifies, generalizes, and directly implies
the two discussed prior classifications.
This trichotomy yields 
 that counting on a set of conjunctive queries
is polynomial-time tractable, is interreducible with the
$\clique$ problem, or admits a reduction from
and is thus as hard as
the counting problem $\sclique$.
In order to prove and state our trichotomy theorem,
we work with a notion of \emph{core} of a 
(structure associated with a) conjunctive query
whose definition crucially takes into account which variables
of the conjunctive query are free
(Definition~\ref{def:augmented}). 
We also use a notion of hypergraph of a conjunctive query
whose vertices are the free variables of the query
(Definition~\ref{def:contract}).
The properties of a query set that determine which case of our trichotomy 
theorem applies are whether or not the cores have bounded treewidth,
and whether or not the just-mentioned hypergraphs have bounded
treewidth;
these two conditions correspond, respectively, to the conditions
that describe the dichotomies for the Boolean case and the
quantifier-free case.
The proof of our trichotomy draws on recent work of
Durand and Mengel~\cite{DurandM14},
who presented a classification for the problem $\sCQ$
based on hypergraphs (see Theorem~\ref{thm:bounded}).

Note that it is readily verified
that the classes of queries for  which we show that $\sCQ$ is
tractable are equivalent to those described in recent work
by Greco and Scarcello~\cite{GrecoS14}.
In contrast to their work,
 we show that the promise version of $\sCQ$ is not only
 fixed-parameter tractable on these classes of queries but can even be
 solved in polynomial time. 
Moreover, and more importantly, we show that these classes are the \emph{only} classes of queries for which $\sCQ$ can be solved efficiently because all other classes of queries are intractable under standard complexity assumptions.

In order to prove and present our trichotomy, we
introduce a version of the \emph{case complexity} framework~\cite{Chen14}
which is suitable for dealing with counting problems.
Among other features, this framework facilitates the presentation
of reductions between parameterized problems which are restricted
in terms of the permitted parameters (or \emph{slices}); 
this is the type of restriction we
deal with here, as the parameter of an instance of $\sCQ$ is
taken to be the query, and we consider $\sCQ$ with respect
to various query sets.
This framework also allows the straightforward derivation
of complexity consequences as a function of the computability
assumption placed on the query sets; witness the derivation
of Theorems~\ref{thm:trichotomy-re} and~\ref{thm:trichotomy-comp} from
Theorem~\ref{thm:trichotomy}.


\section{Preliminaries}

For an integer $i \geq 1$, we define  $\pi_i$ to be
the operator that, given a tuple, returns the value in the $i$th coordinate.

We assume that the reader is familiar with basic graph theoretic notions. In particular, we will use some very basic properties of treewidth, which can be found, for example, in \cite[Chapter 11]{FlumG06} or in \cite[Section 2.3]{thesis}.

\label{sct:homomorphisms}


\subsection{Structures, homomorphisms and cores}

A relational vocabulary is defined to be a set of relation symbols $\tau:=\{R_1, R_2, \ldots, R_\ell\}$ where each $R_i$ has an arity $r_i$. A relational structure $\relA$ over $\tau$ is a tuple $(A, R_1^\relA, \ldots, R_\ell^\relA)$ where $A$ is a set called the \emph{domain}
 of $\relA$ and $R_i^\relA \subseteq A^{r_i}$ is a relation of arity $r_i$. 
All vocabularies and structures in this article are assumed to be relational.
We assume each structure in this article to be \emph{finite}
in that it has a finite domain.
We denote structures by the bold letters, $\relA, \relB, \ldots$,
and their corresponding domains by $A, B, \ldots$.

We assume each class of structures in this article to be of bounded arity, 
that is, for each such class we assume there exists
a constant $c \geq 1$ such that the arity of each relation of a structure in the class
 is at most $c$.  
Since in the bounded arity setting the sizes of all reasonable encodings of a structure are polynomially related, we do not fix a specific encoding  but assume that all structures are encoded in any reasonable way.

\begin{definition}
Let $\relA$ and $\relB$ be two structures over the same vocabulary~$\tau$. A \emph{\myindex{homomorphism}} from $\relA$ to $\relB$ is a function $h:A\rightarrow B$ such that for each relation symbol $R\in \tau$ and each $t=(t_1, \ldots ,t_\ell) \in R^\relA$ we have $(h(t_1), \ldots, h(t_\ell))\in R^\relB$.
A homomorphism $h$ from $\relA$ to $\relB$ is called an \eindex{isomorphism} if $h$ is bijective and $h^{-1}$ is a homomorphism from $\relB$ to $\relA$; when such an isomorphism exists, 
we say that $\relA$ and $\relB$ are \eindex{isomorphic}.
An isomorphism from a structure to itself is called an \eindex{automorphism}.
\end{definition}




\begin{definition}
 Two structures $\relA$ and $\relB$ are \emph{\myindex{homomorphically equivalent}} if there are homomorphisms from $\relA$ to $\relB$ and from $\relB$ to $\relA$.

 A structure is a \emph{\myindex{core}} if it is not homomorphically equivalent to a proper substructure of itself.
 A structure $\relB$ is a \emph{\myindex[core of]{core of a structure}} a structure $\relA$ if
$\relB$ is a substructure of $\relA$,
$\relB$ is homomorphically equivalent to $\relA$, and
$\relB$ is a core.
\end{definition}

We state two basic properties of cores of structures; due to the first, 
we will speak of \emph{the} core of a structure instead of \emph{a} core.
The second seems to be folklore; 
 a proof can be found, for example, in \cite{FlumG06}.
\begin{lemma}\label{lem:coreunique}
 Every structure $\relA$ has at least one core. Furthermore, every two cores $\relB_1$ and $\relB_2$ of $\relA$ are isomorphic.
\end{lemma}

\begin{lemma}\label{lem:ismorphicCores}
 Let $\relA$ and $\relB$ be two homomorphically equivalent structures, and let $\relA'$ and $\relB'$ be cores of $\relA$ and $\relB$, respectively. Then $\relA'$ and $\relB'$ are isomorphic.
\end{lemma}

\subsection{Complexity theory background}


Throughout, we use $\Sigma$ to denote an alphabet over which
strings are formed.
All problems to be considered are viewed as counting problems.
So, a \emph{problem} is a mapping $Q: \Sigma^* \to \N$.
We view decision problems as problems where,
for each $x \in  \Sigma^*$, it holds that $Q(x)$ is equal to $0$ or $1$.
We use $\p$ (as usual) to denote the class of problems 
(which, again, are mappings $\Sigma^* \to \N$)
that can be computed in polynomial time.
A \emph{parameterization} is a mapping
$\kappa: \Sigma^* \to \Sigma^*$.
A parameterized problem is a pair $(Q, \kappa)$
consisting of a problem $Q$ and a parameterization~$\kappa$.
A partial function $T: \Sigma^* \to \N$
is \emph{polynomial-multiplied} 
with respect to a parameterization~$\kappa$
if there exists a computable function $f: \Sigma^* \to \N$
and a polynomial $p: \N \to \N$ such that,
for each $x \in \dom(T)$, 
it holds that $T(x) \leq f(\kappa(x)) p(|x|)$.

\begin{definition}
Let $\kappa: \Sigma^* \to \Sigma^*$
be a parameterization.
A partial mapping $r: \Sigma^* \to \Sigma^*$
is \emph{FPT-computable} with respect to $\kappa$
if there exist a polynomial-multiplied function $T: \Sigma^* \to \N$
(with respect to $\kappa$)
with $\dom(T) = \dom(r)$
and an algorithm $A$ such that,
for each string $x \in \dom(r)$,
the algorithm $A$ computes $r(x)$ within time $T(x)$;
when this holds, we also say that $r$ is \emph{FPT-computable}
with respect to $\kappa$ \emph{via $A$}.
\end{definition}

As is standard, we may and do freely interchange among
elements of $\Sigma^*$, $\Sigma^* \times \Sigma^*$, and $\N$.
We define $\FPT$ to be the class 
that contains a parameterized problem $(Q, \kappa)$
if and only if $Q$ is FPT-computable with respect to $\kappa$.

We now introduce a notion of reduction for counting problems,
which is a form of Turing reduction.
We use $\powfin(A)$ to denote the set containing all finite subsets of
$A$.

\begin{definition}
A \emph{counting FPT-reduction} from a parameterized problem
$(Q, \kappa)$ to a second parameterized problem $(Q', \kappa')$
consists of 
 a computable function $h: \Sigma^* \to \powfin(\Sigma^*)$,
and
an algorithm $A$ such that:
\begin{itemize}
\item  on an input $x$,
$A$ may make oracle queries of the form
$Q'(y)$ with $\kappa'(y) \in  h(\kappa(x))$, and
\item $Q$ is FPT-computable with respect to $\kappa$ via $A$.
\end{itemize}
\end{definition}

We use $\clique$ to denote the decision problem  where $(k, G)$ is a
yes-instance when $G$ is a graph that contains a clique of size $k\in
\mathbb{N}$. By $\sclique$ we denote the problem of counting, given
$(k, G)$, the number of $k$-cliques in the graph $G$.
The parameterized versions of these problems,
denoted by $\pclique$ and $\psclique$, are defined
via the parameterization $\pi_1$.
We will make tacit use of the following well-known facts:
$\FPT$ is closed under counting FPT-reduction;
$\pclique$ is complete for $\W{1}$ under counting FPT-reduction; and,
$\psclique$ complete for $\sW{1}$ under counting
FPT-reduction.

A \emph{promise problem} is a pair $\lpr Q, I \rpr$
where $Q$ is a problem and $I \subseteq \Sigma^*$.
When $C$ is a complexity class, 
define $\prom{C}$ to contain a promise problem
$\lpr Q, I \rpr$ when there exists $P \in C$
such that, for all $x \in I$, it holds that $P(x) = Q(x)$.
A \emph{parameterized promise problem}
is a pair $(\lpr Q, I \rpr, \kappa)$ 
consisting of a promise problem $\lpr Q, I \rpr$ and
a parameterization $\kappa$;
such a problem will also be notated by $\lpr (Q, \kappa), I \rpr$.
When $C$ is a parameterized complexity class,
define $\prom{C}$ to contain a promise problem
$(\lpr Q, I \rpr, \kappa)$ 
when there exists a problem $P$ such that $(P, \kappa) \in C$
and 
for all $x \in I$, it holds that $P(x) = Q(x)$.

\section{Conjunctive Queries and Computational Problems}

 A \emph{conjunctive query} is a relational first-order formula (possibly with free variables) of the form $\exists v_1 \ldots \exists v_n \bigwedge_{i=1}^m \alpha_i$ where each $\alpha_i$ is a predicate application, that is, an atomic formula of the form $R(\vec{u})$ where $R$ is a relation symbol and $\vec{u}$ is a tuple of variables.
Since the only type of queries that we are concerned with in this article are conjunctive queries, we will sometimes simply use \emph{query} to refer to a conjunctive query.

\begin{definition}
\label{def:natural-model}
To a conjunctive query $\phi$ over the vocabulary $\tau$,
 we assign a structure $\relA=\relA_\phi$,
 called the \emph{\myindex{natural model}},
 as follows: 
the domain of $\relA$ is $\var{\phi}$;
$\relA$ is over the vocabulary $\tau$; and
for each relation symbol $R\in \tau$,
 we set $R^\relA:=\{\vec{a} \mid R(\vec{a}) \text{ is an atom of } \phi\}$.
\end{definition}


To each conjunctive query $\phi$ we assign the pair $(\relA,S)$ where $\relA$ is the natural model of $\phi$ and $S$ the set of its free variables. 
From such a pair $(\relA,S)$,
it is easy to reconstruct the corresponding query $\phi$: 
each tuple of a relation of $\relA$ 
is made into an atom, and
then, one existentially quantifies the elements of $A$
 not in $S$ to obtain $\phi$.
Because of this easy correspondence between queries and pairs $(\relA,
S)$, in a slight abuse of notation, we do not
differentiate between pairs $(\relA, S)$ and queries throughout.
 In particular, we will call a pair $(\relA, S)$ a query, and we will use $\calC$ interchangably for classes of queries and of pairs $(\relA,S)$. 

Let $\phi$ be a conjunctive query with assigned pair $(\relA, S)$ and let $\relB$ be a structure. Then a function $h:S \rightarrow B$ is a satisfying assignment of $\phi$ if and only if it can be extended to a homomorphism from $\relA$ to $\relB$;
we denote the set of such functions by $\hom(\relA, \relB, S)$. 
In this article, we are interested in
the following counting problem.

\problem{$\sCQ$}{A query $(\relA,S)$ and a structure $\relB$.}{Compute $|\hom(\relA,\relB,S)|$.}

In the case of a conjunctive query $(\relA,\emptyset)$ without free variables, the problem $\sCQ$ amounts to deciding whether or not there exists a homomorphism from $\relA$ to $\relB$. 
We define this case as the problem $\CQ$.

\problem{$\CQ$}{A query $(\relA,\emptyset)$ and a structure
  $\relB$.}{Decide if
there exists a homomorphism from $\relA$ to $\relB$.}

We define $p$-$\sCQ$ to be the parameterized problem 
$(\sCQ, \pi_1)$, that is, we take the parameter of each instance
$((\relA,S), \relB)$ to be $(\relA,S)$.
(Formally, we view each instance of  $p$-$\sCQ$ 
as a pair of strings, where the first component
encodes the query, and the second component encodes the structure.)
In analogy to $\sCQ$, we define 
$\smallp\CQ$ to be the parameterized problem $(\CQ, \pi_1)$.


%
%

%
%


We define the hypergraph of a query $(\relA,S)$ to be the hypergraph $\calH=(V,E)$ where $V$ is the domain of $\relA$ and $E:= \{ \dom(t)\mid t\in R^\relA, \text{ $R^\relA$ is a relation of $\relA$}\}$ where $\dom(t)$ denotes the set of elements appearing in $t$. The treewidth of $(\relA, S)$ is defined to be that of its hypergraph.
Dalmau, Kolaitis and Vardi \cite{DalmauKV2002} proved that $\CQ$ can be solved efficiently, when the treewidth of the cores of the queries is bounded.

\begin{theorem}[\cite{DalmauKV2002}]\label{thm:DalmauKV}
Let $k\in \mathbb{N}$ be a fixed constant. Let $\calC_k$ be the class of all structures with cores of treewidth at most $k$. Then 
the promise problem 
$\lpr \CQ, \calC_k \times \Sigma^* \rpr$ is in $\prom{\p}$.
\end{theorem}

Grohe \cite{Grohe07} showed that this result is optimal.

\begin{theorem}[\cite{Grohe07}]\label{thm:grohe2}
Let $\calC$ be a recursively enumerable class of structures of bounded arity. Assume $\FPT \ne \W{1}$. Then the following statements are equivalent:
\begin{enumerate}
 \item $\lpr \CQ, \calC \times \Sigma^* \rpr \in \prom{\p}$.
 \item $\lpr  \smallp\CQ,  \calC \times \Sigma^* \rpr   \in \prom{\FPT}$.
 \item There is a constant $c$ such that the cores of the structures in $\calC$ have treewidth at most~$c$.
\end{enumerate}
\end{theorem}

%


Dalmau and Jonsson \cite{DalmauJ04} considered the analogous question for $\sCQ$ for quantifier free queries and found that cores do not help in this setting.

\begin{theorem}[\cite{DalmauJ04}]\label{thm:dalmauJ}
Let $\calQ$ be the class of all quantifier free conjunctive queries,
i.e., queries of the form $(\relA, A)$. 
Let $\calC$ be a recursively enumerable class of structures of bounded arity in $\calQ$. Assume $\FPT \ne \sW{1}$. Then the following statements are equivalent:
\begin{enumerate}
 \item $\lpr \sCQ, \calC \times \Sigma^* \rpr \in \prom{\p}$.
 \item $\lpr \smallp\sCQ, \calC \times \Sigma^* \rpr \in \prom{\FPT}$.
 \item There is a constant $c$ such that the structures in $\calC$ have treewidth at most~$c$.
\end{enumerate}
\end{theorem}

It is common to consider classes of queries defined by restricting their associated hypergraph. For $\sCQ$ it turns out to be helpful to also encode which vertices of a hypergraph correspond to free variables, which is formalized in the following definition. A pair $(\calH, S)$ where $\calH$ is a hypergraph and $S$ is a subset of the vertices of $\calH$ is called an $S$-hypergraph. The $S$-hypergraph of a query $(\relA, S)$ is $(\calH, S)$ where $\calH$ is the hypergraph of $\relA$.

In \cite{DurandM14}, the following version of $\sCQ$ is considered.

\problem{$\sCQh$}{An $S$-hypergraph $(\calH, S)$ and an instance $((\relA, S), \relB)$ of $\sCQ$ where $\calH$ is the hypergraph of $\relA$.}{Compute $|\hom(\relA,\relB,S)|$.}

We define
$p$-$\sCQh$ to be the parameterized problem
$(\sCQh, \pi_1)$; here, 
an instance of $\sCQh$ is viewed as a pair
$((\calH, S), ((\relA, S),\relB))$, on which 
the operator $\pi_1$
returns
$(\calH, S)$.

It turns out that for $\sCQh$ a parameter called $S$-star size is of
critical importance. Let $\calH=(V,E)$ be a hypergraph and $S\subseteq
V$. Let $C$ be the vertex set of a connected component of
$\calH[V-S]$. Let $E_{C}$ be the set of hyperedges $\{e\in E\mid e
\cap C \ne \emptyset\}$ and $V_C:= \bigcup_{e\in E_C}
e$.
Then
$\calH[V_C]$ is called an \emph{$S$-component} of $\calH$. 
The size of a biggest independent set in $\calH[V_C\cap S]$ is called
the $S$-star size of the $S$-component $\calH[V_C]$.
The $S$-star size of $(\calH, S)$ is then defined to be the maximum
$S$-star size taken over all $S$-components of $(\calH, S)$. 
By the \emph{quantified star size} of a query $(\relA, S)$ 
we refer to the $S$-star size of the $S$-hypergraph associated to $(\relA, S)$. For more explanations on these notions and examples see \cite[Section 3.2]{thesis}.

\begin{theorem}[\cite{DurandM14}]\label{thm:bounded}
 Let $\calG$ be a recursively enumerable class of $S$-hypergraphs of bounded arity. Assume that $\W{1} \ne \FPT$. Then the following statements are equivalent:
 \begin{enumerate}
  \item \label{bnd:1} 
$\lpr \sCQh, \calG \times \Sigma^* \rpr \in \prom{\p}$.
  \item \label{bnd:2} 
$\lpr \textup{$p$-$\sCQh$}, \calG \times \Sigma^* \rpr \in \prom{\FPT}$.
  \item \label{bnd:3} There is a constant $c$ such that for each $S$-hypergraph $(\calH,S)$ in $\calG$ 
  the treewidth of $\calH$ and the $S$-star size of $\calH$ 
 are at most $c$.
\end{enumerate}
 \end{theorem}

We have seen that for the problem $\CQ$,
 cores of structures are crucial, while 
in the classification due to Dalmau and Jonsson,
they do not matter at all. Thus we introduce a notion of cores for conjunctive queries that interpolates between these two extreme cases. The idea behind the definition is that we require the homomorphisms between $(\relA, S)$ and its core to be the identity on the free variables, while they may map the quantified variables in any way that leads to a homomorphism. This is formalized as follows.
 
\begin{definition}
\label{def:augmented}
For a conjunctive query $(\relA,S)$ where $\relA$ is defined on vocabulary $\tau$,
we define the \emph{\myindex{augmented structure}},
denoted by
$\aug{\relA,S}$,
to be the structure
 over the vocabulary $\tau\cup \{R_a\mid a\in S\}$ 
where $R_a^{\aug{\relA,S}}:=\{a\}$. 
We define the \emph{\myindex[core]{core of a conjunctive query}} of $(\relA,S)$ 
to be the core of~$\aug{\relA,S}$.
 \end{definition}

\begin{example}
 Let $(\relA,S)$ be a query without free variables, that is, where $S=\emptyset$; then the core of $(\relA, S)$ is the core of $\relA$.
 If $(\relA, S)$ is quantifier-free, that is, where $S=A$, then
 the core of $(\relA, S)$ equals~$\relA$.
\end{example}

The cores of conjunctive queries were essentially already studied by Chandra and Merlin in a seminal paper \cite{ChandraM1977} although the notation used there is different. We give a fundamental result on conjunctive queries.
 We call two queries $(\relA_1, S)$ and $(\relA_2,S)$ \emph{\myindex[equivalent]{equivalent queries}} if for each structure $\relB$ we have $\hom(\relA_1, \relB,S)= \hom(\relA_2, \relB, S)$.

\begin{theorem}[\cite{ChandraM1977}]\label{thm:ChandraMerlin}
 If two conjunctive queries $(\relA_1,S)$ and $(\relA_2, S)$ have the same core
(up to isomorphism), then they are equivalent.
\end{theorem}

\section{Case Complexity}\label{sct:casecomplexity}

In this section we develop a version of the case complexity framework
advocated in \cite{Chen14} which is 
suitable for classifying counting problems.
A main motivation for this framework is the growing amount of research on parameterized problems which are restricted by the permitted values of the parameter. In particular, this kind of problem arises naturally in query answering problems where one often restricts the admissible queries for the inputs (see e.g.~\cite{Grohe07,DalmauJ04,Chen14}). An aim of the case complexity framework as introduced in \cite{Chen14} is to facilitate reductions between  the considered restricted parameterized problems and to show results independent of computability assumptions for the parameter. 

The central notion for our framework is the following: A \emph{case problem} consists of a problem
$Q: \Sigma^* \times \Sigma^* \to \N$
and a subset $S \subseteq \Sigma^*$,
and is denoted $Q[S]$.
When $Q[S]$ is a case problem, we define the following:
\begin{itemize}

\item $\param{Q[S]}$ is the parameterized problem $(P, \pi_1)$
where $P(s, x)$ is defined as equal to 
$Q(s, x)$ if $s \in S$, and as $0$ otherwise.

\item $\prom{Q[S]}$ is the promise problem $\lpr Q, S \times \Sigma^* \rpr$.

\item $\paramprom{Q[S]}$ is the parameterized promise problem
$(\prom{Q[S]}, \pi_1)$.

\end{itemize}

The case problem we consider in this paper will nearly exclusively be $\sCQ[\calC]$ where $\calC$ is a class a class of conjunctive queries. Nevertheless, we stress the fact that our framework is fully generic and we believe that it will in the future also be useful for presenting and proving complexity classifications for other problems.

We now introduce a reduction notion for case problems.

\begin{definition}
A \emph{counting slice reduction} from a case problem $Q[S]$ to a second case problem $Q'[S']$ consists of
\begin{itemize}
 \item a computably enumerable language $U\subseteq \Sigma^* \times \powfin(\Sigma^*)$, and
 \item a partial function 
$r: \Sigma^* \times \powfin(\Sigma^*) \times \Sigma^*\rightarrow \Sigma^*$ 
that has domain $U\times \Sigma^*$ and is 
FPT-computable with respect to $(\pi_1, \pi_2)$ via
an algorithm $A$ that, on input $(s, T, y)$, 
may make queries of the form $Q'(t,z)$
where $t \in T$,
\end{itemize}
such that the following conditions hold:
\begin{itemize}
 \item (coverage) for each $s\in S$, there exists $T \subseteq S'$
such that $(s, T)\in U$, and 
 \item (correctness) for each $(s, T) \in U$, it holds (for each $y\in \Sigma^*$) that 
 \[Q(s,y) = r(s, T, y).\]
\end{itemize}
\end{definition}

As usual in counting complexity, it will often not be necessary to use the full generality of counting slice reductions. Therefore, we introduce a second, parsimonious notion of reductions for case problems which is often general enough but easier to deal with.

\begin{definition}
A \emph{parsimonious slice reduction} from a case problem $Q[S]$ to a second case problem $Q'[S']$ consists of
\begin{itemize}
 \item a computably enumerable language $U\subseteq \Sigma^* \times \Sigma^*$, and
 \item a partial function 
$r: \Sigma^* \times \Sigma^* \times \Sigma^*\rightarrow \Sigma^*$ 
that has domain $U\times \Sigma^*$ and is 
FPT-computable with respect to $(\pi_1, \pi_2)$ 
\end{itemize}
such that the following conditions hold:
\begin{itemize}
 \item (coverage) for each $s\in S$, there exists $s'\in S'$ such that $(s, s')\in U$, and 
 \item (correctness) for each $(t, t') \in U$, it holds (for each $y\in \Sigma^*$) that 
 \[Q(t,y) = Q'(t', r(t, t', y)).\]
\end{itemize}
\end{definition}


We give some basic properties of counting slice reductions. 
Their proofs can be found in the full version of this paper.

\begin{proposition}
If $Q[S]$ parsimoniously slice reduces to $Q'[S']$, then 
 $Q[S]$ counting slice reduces to $Q'[S']$.
\end{proposition}

\begin{theorem} \label{thm:transitivity}
Counting slice reducibility is transitive.
\end{theorem}

The next two theorems give the connection between case complexity and parameterized complexity. 
In particular, they
show that, from a counting slice reduction,
one can obtain 
complexity results for the corresponding
parameterized problems.

\begin{theorem}
\label{thm:slice-red-gives-fpt-red}
Let $Q[S]$ and $Q'[S']$ be case problems.
Suppose that $Q[S]$ counting slice reduces to $Q'[S']$,
and that both $S$ and $S'$ are computable.
Then $\param{Q[S]}$ counting FPT-reduces to $\param{Q'[S']}$.
\end{theorem}


\begin{theorem}
\label{thm:promfpt-to-fpt}
Let $Q[S]$ be a case problem, and let
$K: \Sigma^* \times \Sigma^* \to \N$ be a problem.
Suppose that $\paramprom{Q[S]}$ is in $\prom{\FPT}$,
$S$ is computably enumerable, and that
the case problem $K[\Sigma^*]$ counting slice reduces to $Q[S]$.
Then the parameterized problem $(K, \pi_1)$ is in $\FPT$.
\end{theorem}

In the remainder of the paper, we will show all our reductions in the case complexity framework and then use Theorem~\ref{thm:slice-red-gives-fpt-red} and Theorem~\ref{thm:promfpt-to-fpt} to derive parameterized complexity results. This approach lets us give results on $\sCQ[\calC]$ for different complexity assumptions on $\calC$ without having to deal with these assumptions in the proofs. Thus we separate the technicalities of the reductions from the assumptions on $\calC$ which in our opinion gives a far clearer presentation.


%
%
%

\section{Statement of the main results}

In this section we present the main results of this paper which we will then prove in the remainder of the paper. For the statement of the results we will use certain $S$-hypergraphs that we get as a contraction of the $S$-hypergraphs of conjunctive queries. When deleting a vertex $v$ from a hypergraph, we delete $v$ from the vertex set and all edges it appears in but keep all edges, unless they become empty after the deletion of $v$.

\begin{definition}
\label{def:contract}
To every $S$-hypergraph $(\calH,S)$ we define an $S$-hypergraph $\contract(\calH,S)$ as follows: 
We add an edge $\{u,v\}$
 for any pair of vertices $u,v$ that appears in a common $S$-component of $\calH$. Then we delete the vertices in $V(\calH)\setminus S$. To a class $\calG$ of $S$-hypergraphs we define $\contract(\calG):= \{\contract(\calH,S)\mid (\calH,S) \in \calG\}$.

For a conjunctive query $(\relA, S)$ let $\contract(\relA, S)$ be $\contract(\calH,S)$ where $(\calH,S)$ is the $S$-hypergraph of the core of $(\relA, S)$. For a class $\calC$ of conjunctive queries, set $\contract(\calC):=\{ \contract(\relA, S)\mid (\relA, S)\in \calC\}$.
\end{definition}

We first present a version of our main result using the framework of case complexity.

\begin{theorem}
\label{thm:trichotomy}
Let $\calC$ be a class of conjunctive queries. 
\begin{enumerate}
 
\item If the cores of $\calC$ and $\contract(\calC)$ are of bounded treewidth, then 
$\prom{\sCQ[\calC]}\in \prom{\p}$, and hence
$\param{\prom{\sCQ[\calC]}}  \in   \prom{\FPT}$.

 \item If the cores of $\calC$ are of unbounded treewidth but $\contract(\calC)$ is of bounded treewidth, then $\sCQ[\calC]$ is equivalent to $\clique[\mathbb{N}]$ with respect to counting slice reductions.

 \item If the treewidth of $\contract(\calC)$ is unbounded, then there is a counting slice reduction from $\sclique[\mathbb{N}]$ to $\sCQ[\calC]$.

\end{enumerate}
\end{theorem}

We will prove Theorem~\ref{thm:trichotomy} in Section \ref{sct:puttingtogether}.
Using the results on case complexity, we derive from Theorem~\ref{thm:trichotomy} two versions of the trichotomy
phrased in terms of promise problems and 
of non-promise problems, depending on whether or not
the class $\calC$ of conjunctive queries is assumed
to be recursively enumerable or computable, respectively.

\begin{theorem}
\label{thm:trichotomy-re}
Let $\calC$ be a class of conjunctive queries which is
recursively enumerable.
In the scope of this theorem, let us say that the class $\calC$
is \emph{tractable} if 
$\lpr \sCQ, \calC \times \Sigma^* \rpr \in \prom{\p}$ and
$\lpr \smallp\sCQ, \calC \times \Sigma^* \rpr \in \prom{\FPT}$.

\begin{enumerate}

\item If the cores of $\calC$ and $\contract(\calC)$
have bounded treewidth, then 
$\calC$ is tractable.

\item If the cores of $\calC$ have unbounded treewidth,
then $\calC$ is not tractable, unless $\pclique$ is in FPT
(and hence $\FPT = \W{1}$).

\item If $\contract(\calC)$ has unbounded treewidth,
then $\calC$ is not tractable, unless $\psclique$ is in FPT
(and hence $\FPT = \sW{1}$).

\end{enumerate}
\end{theorem}

\begin{proof}
(1) follows directly from item (1) of Theorem~\ref{thm:trichotomy}.
(2) and (3)
follow directly from the respective items of 
Theorem~\ref{thm:trichotomy}
and Theorem~\ref{thm:promfpt-to-fpt}.
\end{proof}

From Theorem~\ref{thm:trichotomy-re},
one can immediately derive, as 
corollaries, Theorem~\ref{thm:grohe2}
and Theorem~\ref{thm:dalmauJ}.

\newcommand{\res}{\upharpoonright}

Let us use $(\smallp\sCQ \res I)$ 
to denote the parameterized problem which 
is equal
to $\sCQ$ on $I$, and is equal to $0$ elsewhere
(and has the
parameterization of $\smallp\sCQ$).

\begin{theorem}
\label{thm:trichotomy-comp}
Let $\calC$ be a class of conjunctive queries which is
computable.
\begin{enumerate}

\item  If the cores of $\calC$ and $\contract(\calC)$
have bounded treewidth, then 
$(\smallp\sCQ \res \calC \times \Sigma^*)$ is in FPT.

\item If the cores of $\calC$ have unbounded treewidth,
and $\contract(\calC)$ has bounded treewidth, 
then
$(\smallp\sCQ \res \calC \times \Sigma^*)$ is equivalent to 
$\pclique$ under counting FPT-reduction.

\item If $\contract(\calC)$ has unbounded treewidth,
$(\smallp\sCQ \res \calC \times \Sigma^*)$ admits
a counting FPT-reduction from $\psclique$.

\end{enumerate}
\end{theorem}

\begin{proof}
For (1), the FPT algorithm is to first decide, given an instance $(\phi, \relB)$,
whether or not $\phi \in \calC$; if so, the algorithm invokes
the algorithm of Theorem~\ref{thm:trichotomy},
otherwise, it returns $0$.
(2) and (3) follow immediately from Theorem~\ref{thm:trichotomy}
and Theorem~\ref{thm:slice-red-gives-fpt-red}.
\end{proof}

\section{Positive complexity results}\label{sct:algorithms}

In this section,
we will prove a counting version of Theorem~\ref{thm:DalmauKV}. 
We will use a lemma that is probably well known, but as we could not find a reference, we give a proof for it\iftoggle{ARXIV}{}{ in Appendix~\ref{app:computecores}}.

\begin{lemma}\label{lem:computecores}
Let $k\in \mathbb{N}$ be a fixed constant. There exists a polynomial-time algorithm that, given a structure $\relA$ whose core has treewidth at most $k$, outputs a core of $\relA$.
\end{lemma}
\def\prooflemcomputecores{
 The proof is based on well-known query minimization techniques already pioneered in \cite{ChandraM1977}. The basic observation is that if~$\relA$ is not a core, then there is a substructure $\relA_s$, that we get by deleting a tuple from a relation of $\relA$, that contains a core of $\relA$. Trivially,~$\relA_s$ is homomorphically equivalent to $\relA$. Thus by Lemma \ref{lem:ismorphicCores}, for every substructure $\relA_s$ of $\relA$ that is homomorphically equivalent to $\relA$, the core of $\relA_s$ is also a core of $\relA$.
 
 The construction of a core $\relA_c$ goes as follows: For every tuple $t$ in every relation check, using the algorithm of Theorem~\ref{thm:DalmauKV}, if the structure we get from $\relA$ by deleting $t$ is homomorphically equivalent to $\relA$. If there is such a tuple $t$, delete it from $\relA$ and iterate the process, until no tuple can be deleted anymore. 
 
 By the discussion above the end result $\relA_c$ of this procedure must be a core. Furthermore, $\relA_c$ is homomorphically equivalent to $\relA$ and a substructure of $\relA$, so $\relA_c$ is a core of $\relA$. Finally, at most $\|\relA\|$ tuples get deleted and for every deleted tuple the algorithm has to perform at most $\|\relA\|$ homomorphism tests. The left hand sides of these tests all have the same core of treewidth at most $k$ and the right hand sides have size at most $\|\relA\|$. Using Theorem \ref{thm:DalmauKV} then gives a runtime polynomial in $\|\relA\|$.
\end{proof}
}
\iftoggle{ARXIV}{
\begin{proof}
\prooflemcomputecores
}{}
Lemma~\ref{lem:computecores} yields a counting version of Theorem~\ref{thm:DalmauKV} as an easy corollary.

\begin{corollary}\label{cor:countingalg}
 Let $\calC$ be a class of conjunctive queries such that the cores of the queries in~$\calC$ have bounded quantified star size and bounded treewidth. Then $\prom{\sCQ[\calC]}\in \prom{\p}$.
\end{corollary}
\begin{proof}
Let $((\relA,S), \relB)$ be an instance of $\prom{\sCQ[\calC]}$ with domain $A$. By the promise, there is a constant $c$ such that the treewidth and the quantified star size of the core of $(\relA,S)$ are at most $c$. We simply compute the core of $(\relA,S)$ with Lemma \ref{lem:computecores} and delete from it the relations $R_a$ introduced when constructing $\aug{\relA}$. Call the resulting query $(\hatrelA, S)$. By construction, $(\relA, S)$ and $(\hatrelA, S)$ have the same core, so by Theorem \ref{thm:ChandraMerlin} are equivalent. Moreover, the treewidth and quantified star size of $(\hatrelA,S)$ are bounded by $c$ and thus Theorem \ref{thm:bounded} lets us solve the instance in polynomial time.
\end{proof}

Let us discuss Corollary~\ref{cor:countingalg}. 

Theorem~\ref{thm:DalmauKV} and thus also Lemma~\ref{lem:computecores} crucially depends on the fact that we know by an outside promise that the treewidth of the cores we consider is bounded. If this bound is not satisfied, then the algorithm of Theorem~\ref{thm:DalmauKV} may give false positive results. Consequently, the algorithm of Lemma~\ref{lem:computecores} may compute a structure that is in fact \emph{not} the core of the input and then the algorithm of Corollary~\ref{cor:countingalg} gives the wrong count. Unfortunately, deciding if the core of a conjunctive query has treewidth at most $k$ is $\NP$-complete~\cite{DalmauKV2002} and even the problem of deciding if a fixed structure is the core of a given structure is $\NP$-complete. Thus there is no efficient way of realizing that the core computed by the algorithm of Lemma~\ref{lem:computecores} is wrong.

Consequently, while the result of Corollary~\ref{cor:countingalg} is very nice from a theoretical point of view (we will see in the next section that it is in fact optimal), it is probably of limited value from a more practical perspective. We see this as evidence that in fact parameterized complexity is a framework better suited for the type of problem discussed in this paper. 
Note that in this more relaxed setting of parameterized complexity, 
computing the core of a query by brute force can easily be done in the allowed time, because the core depends only on the query which is the parameter.

\newcommand{\all}{\mathsf{ALL}}

We now present a counting algorithm for $\sCQ[\calC]$ for certain classes $\calC$ that has oracle access to $\CQ$, the decision version of $\sCQ$. Let $\all$ be the class of all conjunctive queries.

\begin{lemma}\label{lem:decision}
 Let $\calC$ be a class of queries such that the treewidth of the $S$-hypergraphs in $\contract(\calC)$ is bounded by a constant $c$. Then there is a counting slice reduction from $\sCQ[\calC]$ to $\CQ[\all]$.
\end{lemma}

The idea of the proof is as follows: Since the treewidth of $\contract(\calC)$ is bounded, we know that the unbounded treewidth of the cores does not originate from the structure of the free variables but only from the way the quantified variables interact in the $S$-components. We use the oracle for $\CQ[\all]$ to solve subqueries of the original query in order to ``contract'' the quantified variables into one variable per $S$-component. This results in an instance with the same solutions that has bounded treewidth. We then solve this instance with the algorithm of Theorem~\ref{thm:bounded}. \iftoggle{ARXIV}{We now give the full proof of Lemma~\ref{lem:decision}.}{The complete proof of Lemma~\ref{lem:decision} can be found in Appendix~\ref{app:decision}.}

\def\prooflemdecision{
\begin{proof}
 The construction is similar to that in the proof of Lemma~\ref{lem:contracthard}. Let $U$ be the relation that relates, to a query $(\relA, S)$ over $\tau$, 
the set of all queries of the form $(\relA', A')$ over $\tau$ with $|A'| \le |A|$.
 
 We now describe how the algorithm for $r$ works, given $(\relA, S)$, the queries related to $(\relA, S)$ by $U$, and $\relB$. W.l.o.g.~we assume that $\aug{\relA, S}$ is a core. Let $\calH$ be the hypergraph of $\relA$. We assume w.l.o.g.~that for every edge $e$ of $\calH$ the structure $\relA$ contains one relation~$R_e^\relA$ with only a single tuple $\vec{e}$ where $\vec{e}$ contains the elements of $e$ in an arbitrary order. We compute a new structure $\relA'$ as follows: We add for each $S$-component of $\relA$ a new domain element $a_C$. Then we add for each component $C$ and each element $a\in S\cap V(C)$ a new relation $R_{C,a}:=\{(a, a_C)\}$. Finally, we delete $A\setminus S$.
 
 If is easy to verify that the $S$-hypergraph of $(\relA,S)$ has treewidth at most $c+1$.
 
 We now construct a structure $\relB'$ over the same vocabulary as $\relA'$. For the relation symbols already present in $\relA$, we set $R^{\relB'} := R^\relB[S]$. For every $S$-component $C$ of $(\calH, S)$, let $D_C$ be the tuples encoding the elements of $\hom(\relA[V(C)], \relB,S\cap V(C))$. For every $v\in V(C)\setminus S$ we let $D_C$ be the domain of $v$. We let $R_{a,C}^{\relB'}$ contain the pairs $(u,v)$ such that $u$ coincides with the assignment to $a$ that is encoded in the assignment to $v$.
 
 We claim that the construction of $((\relA', S), \relB')$ can be done by an $\FPT$-algorithm with the given oracle. First note that, since the treewidth of $\contract(\calH)$ is bounded by $c$, there can be no clique of size greater than $c+1$ in $\contract(\calH)$. It follows that every $S$-component of $(\calH, S)$ can have at most $c+1$ vertices from $S$. Thus the size of $D_C$ is bounded by a polynomial in $A$. Moreover, for every mapping $h:V(C) \cap S \mapsto B$ it can be checked with the oracle if $h$ is to be added to $D_C$. Thus the procedure so far is an $\FPT$-algorithm.
 
 It is easy to see that $\hom(\relA, \relB, S) = \hom(\relA', \relB', S)$. But $(\relA',S)$ is of treewidth at most $c+1$ and its $S$-hypergraph has $S$-star size at most $c+1$. Thus we can compute $|\hom(\relA', \relB', S)|$ in polynomial time by Theorem~\ref{thm:bounded} which completes the proof.
\end{proof}
}

\iftoggle{ARXIV}{\prooflemdecision}{}

\section{Hardness results}\label{sct:hardness}

In this section we will prove the hardness results for Theorem~\ref{thm:trichotomy}. The main idea is reducing from the hard cases of Theorem~\ref{thm:bounded} in several steps.

\subsection{Simulating unary relations}

In this section we show that for queries whose augmented structure is a core we can simulate unary relations on the variables of the query. These additional relations will later allow us to tell the variables apart such that we can later simulate the case in which all atoms of the queries have different relation symbols.

 \begin{lemma}\label{lem:bijection}
  Let $(\relA,S)$ be a conjunctive query such that $\aug{\relA,S}$ is a core. Then every homomorphism $h:\relA\rightarrow \relA$ with $h|_S = \id$ is a bijection.
 \end{lemma}
\begin{proof}
 Clearly, $h$ is also a homomorphism $h:\aug{\relA,S}\rightarrow \aug{\relA,S}$, because $h(a)=a\in R_a^{\aug{\relA,S}}$ for every $a\in S$. But by assumption $\aug{\relA,S}$ is a core, so there is no homomorphism from $\aug{\relA,S}$ to a proper substructure and thus $h$ must be a bijection on $\aug{\relA,S}$ and consequently also on $\relA$.
\end{proof}

We assign a structure $\relA^*$ to every structure $\relA$: 

\begin{definition}
To a structure $\relA$ we assign the structure $\relA^*$ over the vocabulary $\tau\cup \{R_a\mid a\in A\} $ defined as $\relA^*:= \relA\cup \bigcup_{a\in A} R_a^{\relA^*}$ where $R_a^{\relA^*}:=\{a\}$. 
\end{definition}

Note that $\aug{\relA,S}$ and $\relA^*$ differ in which relations we add: For the structure $\aug{\relA,S}$ we add~$R_a^{\aug{\relA,S}}$ for variables $a\in S$ while for $\relA^*$ we add $R_a^{\relA^*}$ for all $a\in A$. Thus,~$\relA^*$ in general may have more relations than $\aug{\relA,S}$.

We now formulate the main lemma of this section whose proof uses ideas from \cite{DalmauJ04}.

\begin{lemma}\label{lem:constants}
Let $\calC$ be a class of conjunctive queries such that for each $(\relA, S)\in \calC$ the augmented structure $\aug{\relA,S}$ is a core. Let $\calC^*:=\{(\relA^*, S)\mid (\relA, S)\in \calC\}$. Then there is a counting slice reduction from $p$-$\sCQ[\calC^*]$ to $p$-$\sCQ[\calC]$.
\end{lemma}
\begin{proof}
 Let $((\relA^*,S),\relB)$ be an input for $\sCQ[\calC^*]$. Remember that $\relA^*$ and $\relB$ are structures over the vocabulary $\tau \cup \{R_a\mid a\in A\}$. 
For every query $(\relA, S)$, the relation $U$ of our counting slice reduction contains $((\relA^*, S), (\relA, S))$.
Obviously, $U$ is computable and satisfies the coverage property.

 We now will reduce the computation of the size $|\hom(\relA^*, \relB, S)|$ to the computation of $|\hom(\relA, \relB', S)|$ for different structures $\relB'$.
 
 Let $D:= \{(a,b)\in A\times B\mid b \in R_a^\relB\}$ and define a structure $\relD$ over the vocabulary $\tau$ with the domain $D$ that contains for each relation symbol $R\in \tau$ the relation \begin{align*}R^\relD:= \{ ((a_1, b_1), \ldots , (a_r, b_r)) \mid &(a_1, \ldots , a_r)\in R^\relA, (b_1, \ldots , b_r)\in R^\relB, \\&\forall i\in [r]: (a_i, b_i)\in D\}.\end{align*} Let again $\pi_1: D\rightarrow A$ be the projection onto the first coordinate, i.e., $\pi_1(a,b):=a$.
 Observe that $\pi_1$ is by construction of $\relD$ a homomorphism from $\relD$ to $\relA$.
 
 We will several times use the following claim:
 
 \begin{claim}\label{clm:automorphism}
  Let $h$ be a homomorphism from $\relA$ to $\relD$ with $h(S)=S$. Then $\pi_1\circ h$ is an automorphism of $\relA$.
 \end{claim}
\begin{proof}
 Let $g:=\pi_1\circ h$. As the composition of two homomorphisms, $g$ is a homomorphism from $\relA$ to $\relA$. Furthermore, by assumption $g|_S$ is a bijection from $S$ to $S$. Since $S$ is finite, there is $i\in \mathbb{N}$ such that $g^i|_S= \id$. But $g^i$ is a homomorphism and thus, by Lemma \ref{lem:bijection}, $g^i$ is a bijection. It follows that $g$ is a bijection. 
 
Since $A$ is finite, there is $j\in \mathbb{N}$ such that $g^{-1}= g^l$. It follows that~$g^{-1}$ is a homomorphism and thus $g$ is an automporphism.
\end{proof}

 Let $\calN$ be the set of mappings $h:S\rightarrow D$ with $\pi_1 \circ h = \id$ that can be extended to a homomorphism $h':\relA\rightarrow \relD$.\hubie{can a more descriptive letter be used, also?}\stefan{I have no idea what a more descriptive letter might be}
 
 \begin{claim}\label{clm:constIdentity}
There is a bijection between $\hom(\relA^*, \relB, S)$ and $\calN$.
 \end{claim}
\begin{proof}
For each $h^*\in \hom(\relA^*, \relB, S)$ we define $P(h^*):=h$ by $h(a):=(a,h^*(a))$ for $a\in S$. From the extension of $h^*$ to $A$ we get an extension of $h$ that is a homomorphism and thus $h\in \calN$. Thus $P$ is a mapping $P:\hom(\relA^*, \relB,S)\rightarrow \calN$.

We claim that $P$ is a bijection. Clearly, $P$ is injective. We we will show that it is surjective as well. To this end, let $h:S\rightarrow D$ be a mapping in $\calN$ and let~$h_e$ be a homomorphism from $\relA$ to $\relD$ that is an extension of $h$. By definition of~$\calN$ such a $h_e$ must exist. By Claim \ref{clm:automorphism} we have that $\pi_1\circ h_e$ is an automorphism, and thus $(\pi_1\circ h_e)^{-1}$ is a homomorphism. We set $h_e' := h_e\circ(\pi_1\circ h_e)^{-1}$. Obviously,~$h_e'$ is a homomorphism from $\relA$ to $\relD$, because $h_e'$ is the composition of two homomorphisms. Furthermore, for all $a\in S$ we have $h_e'(a)= (h_e\circ (\pi_1\circ h_e))(a) = (h_e\circ (\pi_1\circ h))(a) = h_e(a) = h(a)$, so $h_e'$ is an extension of $h$. Moreover $\pi_1\circ h_e' = (\pi_1\circ h_e)\circ (\pi_1\circ h_e)^{-1} = \id$. Hence, we have $h_e' = \id \times \hat{h}$ for a homomorphism $\hat{h}:\relA \rightarrow {\hatrelB}$, where $\hatrelB$ is the structure we get from $\relB$ by deleting the relations~$R_a^{\relB}$ for $a\in A$. But by definition $h_e'(a)\in D$ for all $a\in A$ and thus $\hat{h}(a)\in R_a^\relB$. It follows that $\hat{h}$ is a homomorphism from $\relA^*$ to $\relB$. We set $h^*:=\hat{h}|_S$. Clearly, $h^*\in \hom(\relA^*, \relB,S)$ and $P(h^*) = h$. It follows that $P$ is surjective. This proves the claim.
\end{proof}

Let $I$ be the set of mappings $g:S\rightarrow S$ that can be extended to an automorphism of $\relA$. Let $\calN'$ be the set of mappings $h:S \rightarrow D$ with $(\pi_1 \circ h)(S)=S$ that can be extended to homomorphisms $h':\relA\rightarrow \relD$.

\begin{claim}
 \[|\hom(\relA^*, \relB, S)| = \frac{|\calN'|}{|I|}.\]
\end{claim}
\begin{proof}
Because of Claim \ref{clm:constIdentity} it is sufficient to show that
\begin{equation}\label{eq:0a}
|\calN'| = |\calN||I|. 
\end{equation}
We first prove that
\begin{equation}\label{eq:1a}
\calN' = \{f\circ g \mid f\in \calN, g\in I\}. 
\end{equation}
The $\supseteq$ direction is obvious. For the other direction let $h\in \calN'$. Let~$h'$ be the extension of $h$ that is~a homomorphism $h':\relA\rightarrow \relD$. By Claim \ref{clm:automorphism}, we have that $g := \pi_1\circ h'$ is an automorphism of $\relA$. It follows that $g^{-1}|_S\in I$. Furthermore, $h\circ g^{-1}|_S$ is a mapping from $S$ to~$D$ and $h'\circ g^{-1}$ is an extension that is a homomorphism from $\relA$ to $\relD$. Furthermore $(\pi_1\circ h'\circ g^{-1}|_S)(a) = (g|_S \circ g^{-1}|_S)(a)=a$ for every $a\in S$ and hence $h'\circ g^{-1}|_S \in \calN$ and $h= h \circ g^{-1}|_S \circ g|_S$ which proves the claim~(\ref{eq:1a}).

To show (\ref{eq:0a}), we claim that for every $f,f'\in \calN$
and every $g, g'\in I$, if $f\ne f'$ or $g\ne g'$, then $f\circ g\ne f'\circ g'$. To see this, observe that $f$ can always be written as $f= \id\times f_2$ and thus $(f\circ g)(a) = (g(a), f_2(g(a))$. Thus, if $g$ and $g'$ differ, $\pi_1\circ f\circ g \ne \pi_1 \circ f' \circ g'$ and thus $f\circ g\ne f'\circ g'$. Also, if $g=g'$ and $f\ne f'$, then clearly $f\circ g\ne f'\circ g'$. This completes the proof of (\ref{eq:0a}) and the claim.
\end{proof}

Clearly, the set $I$ depends only on $(\relA,S)$ and thus it can be computed by an $\FPT$-algorithm. Thus it suffices to show how to compute $|\calN'|$ in the remainder of the proof.

For each set $T\subseteq S$ we define $\calN_T:= \{h\in \hom(\relA, \relD, S) \mid (\pi_1\circ h)(S) \subseteq T\}$. We have by inclusion-exclusion 
\begin{equation}\label{eq:2a}
|\calN'| = \sum_{T\subseteq S} (-1)^{|S\setminus T|} |\calN_T|.                                                                                                                        
\end{equation}
Observe that there are only $2^{|S|}$ summands in (\ref{eq:2a}) and thus if we can reduce all of them to $\sCQ$ with the query $(\relA,S)$ this will give us the desired counting slice reduction.

We will now show how to compute the $|\calN_T|$ by interpolation.
So fix a $T\subseteq S$. Let $\calN_{T,i}$ for $i=0,\ldots, |S|$ 
consist of the mappings $h\in \hom(\relA, \relD, S)$ such that there are exactly $i$ elements $a\in S$ that are mapped to $h(a)= (a', b)$ such that $a'\in T$.
Obviously, $\calN_T = \calN_{T,|S|}$ with this notation.

Now for each $j=1,\ldots , |S|$ we construct a new structure $\relD_{j,T}$ over the domain $D_{j,T}$. 
To this end, for each $a \in T$, let $a^{(1)}, \ldots, a^{(j)}$ be copies of $a$
which are not in $D$. Then we set \begin{align*}D_{j,T}:= \{(a^{(k)}, b)\mid (a,b)\in D, a\in T, k\in [j]\} \cup \{(a,b)\mid (a,b)\in D, a\notin T\}.\end{align*} We define a mapping $B:D\rightarrow \wp(D_{j,T})$, where $\wp(D_{j,T})$ is the power set of $D_{j,T}$, by
\[B(a,b):= \begin{cases}
            \{(a^{(k)},b)\mid k\in [j]\}\}, & \text{if } a\in T\\
            \{(a,b)\}, & \text{otherwise}.
           \end{cases}\]
For every relation symbol $R\in \tau$ we define
$R^{\relD_{T,j}} := \bigcup_{(d_1, \ldots, d_s)\in R^{D}} B(d_1)\times \ldots \times B(d_s).$


Then every $h\in \calN_{T,i}$ corresponds to $i^j$ mappings in $\hom(\relA, \relD_{j,T},S)$. Thus for each~$j$ we get 
$\sum_{i=1}^{|S|} i^j |\calN_{T,i}| = |\hom(\relA,\relD_{j,T}, S)|.$
This is a linear system of equations and the corresponding matrix is a Vandermonde matrix, so $\calN_T = \calN_{T,|S|}$ can be computed with an oracle for $\sCQ$ on the instances $((\relA, S),\relD_{j, T})$. The size of the linear system depends only on~$|S|$. Furthermore, $\|D^j\|\le \|D\| j^s\le \|D\|^{s+1}$ where $s$ is the bound on the arity of the relations symbols in $\tau$ and thus a constant. It follows that the algorithm described above is a counting slice reduction. This completes the proof of Lemma~\ref{lem:constants}.
\end{proof}

\subsection{Reducing from hypergraphs to structures}

In this section we show that we can in certain situations reduce from $\sCQh$ to $\sCQ$. This will later allow us to reduce from the hard cases in Theorem~\ref{thm:bounded} to show the hardness results of Theorem~\ref{thm:trichotomy}.

We proceed in several steps.
Let in this section $\calC$ be a class of conjunctive queries of bounded arity. 
To every query $(\relA, S)$ we construct a structure $\hat{\relA}$ as follows; note that when we use this notation, $S$ will be clear from the context. Construct the augmented structure $\aug{\relA,S}$ of $\relA$ and compute its core.
 We define $\hatrelA$ to be the structure that we get by deleting the relations~$R_a$ for $a\in S$ that we added in the construction of $\aug{\relA,S}$. We set $\hatcalC:=\{(\hat{\relA},S) \mid (\relA,S)\in \calC\}.$

Note that in any situation where we apply both the $\hat{ }$ - and ${ }^*$-operators, the $\hat{ }$ is applied before the ${ }^*$.
\begin{claim}
  There is a parsimonious slice-reduction from $\sCQ[\hatcalC]$ to $\sCQ[\calC]$.
 \end{claim}
\begin{proof}
The relation $U$ relates to every query $(\relA, S)$ the query $(\hatrelA, S)$. Certainly, $U$ is computable and by definition assigns to each query in $\hatcalC$ a query in $\calC$.
 
We have that $\hatrelA$ is a substructure of $\relA$ and there is a homomorphism from $\relA$ to $\hatrelA$, because there is a homomorphism from $\aug{\relA,S}$ to $\aug{\hatrelA,S}$.
 Hence, $\relA$ and $\hatrelA$ are homomorphically equivalent and by Theorem \ref{thm:ChandraMerlin} we have that $(\relA,S)$ and $(\hatrelA,S)$ are equivalent. Thus setting $r((\relA,S), (\hatrelA,S), \relB)):= \relB$ yields the desired parsimonious slice-reduction.
\end{proof}
 
Let $\hatcalC^*:=\{(\hatrelA^*, S)\mid (\hatrelA, S)\in \hatcalC\}$. Note that, by Lemma~\ref{lem:constants}, there is a counting slice reduction from $p$-$\sCQ[\hatcalC^*]$ to $p$-$\sCQ[\hatcalC]$.

 Let now $\calG$ be the class of $S$-hypergraphs associated to the queries in $\hatcalC$. 
 
\begin{claim}
 There is a parsimonious slice reduction from $\sCQh[\calG]$ to $\sCQ[\hatcalC^*]$.
\end{claim}
\begin{proof}
The relation $U$ relates every $S$-hypergraph $(\calH,S)$ in $\calG$ to all queries $(\hatrelA^*, S)$ with the hypergraph $(\calH,S)$. Certainly, $U$ is computable and by definition of $\calG$ it assigns to $S$-hypergraph in $\calG$ a query in $\hatcalC^*$.

It remains to describe the function $r$.
So let $((\relA,S),\relB)$ be a $p$-$\sCQ$-instance such that $(\relA,S)$ has the $S$-hypergraph $(\calH, S)\in\calG$. 
 We assume w.l.o.g.~that every tuple appears only in one relation of $\relA$. If this is not the case, say a tuple $t$ appears in two relations $R_1^\relA$ and $R_2^\relA$, then we build a new instance as follows: Delete $t$ from $R_1^\relA$ and $R_2^\relA$, add a new relation $R_t^\relA$ to $\relA$ containing only $t$. Finally, set $R^B=R_1^B\cap R_2^\relB$. This operation does not change the associated $S$-hypergraph, so this new instance still has the $S$-hypergraph $(\calH,S)$. Moreover it is easy to see that it has the same set of solutions.

Let the vocabulary of $(\hatrelA, S)$ be $\tau$.
We construct a structure $r((\relA,S), (\hatrelA^*,s),\relB)=: \hatrelB$ over the same relation symbols as~$\hatrelA^*$, i.e., over the vocabulary $\tau \cup \{R_a\mid a\in \hatA\}$. The structure $\hatrelB$ has the domain $\hatB:=A\times B$ where $A$ is the domain of $\relA$ and $B$ is the domain of $\relB$.
 For $\hatR\in \tau$ we set 
 \begin{align*}\hatR^{\hatrelB}:=\{ ((a_1, b_1), \ldots , (a_k, b_k)) \mid & (a_1, \ldots, a_k)\in \hatR^{\hatrelA}, (a_1, \ldots, a_k)\in R^\relA, \\&(b_1, \ldots, b_k)\in R^\relB\}.\end{align*}
 Furthermore, for the relations symbols $\hatR_a$ that are added in the construction of $\hatrelA^*$ from $\hatrelA$ we set $\hatR_a^{\hatrelB}:= \{(a, b)\mid b\in B\}$, where $B$ is the domain of $\relB$.
 
It is easy to see that from a satisfying assignment $h:\relA \rightarrow \relB$ we get a homomorphism $h':\hatrelA^* \rightarrow \hatrelB$ by setting $h'(a):= (a, h(a))$. Furthermore, this construction is obviously bijective. Thus we get $|\hom(\relA, \relB,S)| = |\hom(\hatrelA^*, \hatrelB, S)|$. 
Since $\hatrelB$ can  be constructed in polynomial time in $\|\relA\|$ and $\|\relB\|$, this is a parsimonious slice reduction.
 \end{proof}

\begin{corollary}\label{cor:graphstostructures}
Let $\calC$ be a class of conjunctive queries of bounded arity and let $\calG$ be the class of $S$-hypergraphs of the cores of $\calC$. Then there is a counting slice reduction from $\sCQh[\calG]$ to $\sCQ[\calC]$.
\end{corollary}

\subsection{Strict star size}

In this section we introduce a notion strict $S$-star size to simplify some of the arguments in the next section.
We define the strict $S$-star size of a hypergraph to be the maximum number of vertices in $S$ that are contained in one $S$-component of $\calH$.
\iftoggle{ARXIV}{
The aim of this section is the following Lemma.
}{}

\begin{lemma}\label{lem:strictSS}
 Let $\calG$ be a class of $S$-hypergraphs of bounded arity. If the strict $S$-star size of the $S$-hypergraphs in $\calG$ is unbounded, then there is a counting slice reduction from $\clique[\mathbb{N}]$ to $\sCQh[\calG]$.
\end{lemma}

\def\proofstrictSS{
We start with some easy observations. 

\begin{observation}\label{obs:primal}
 Let $\calG$ be a class of $S$-hypergraphs of bounded arity. Define $\calG' \{ (\calH', S) \mid (\calH, S)\in \calG, \text{ $\calH'$ is the primal graph of $\calH$}\}$. Then there is a parsimonious slice reduction from $\sCQh[\calG']$ to $\sCQh[\calG]$.
\end{observation}
\begin{proof}
Since the arity of the $S$-hypergraphs in  $\calG$ is bounded by a constant $c$, we can simulate the binary relations on the edges of any primal graph $\calH'$ by the relations for $\calH$. Note that every relation for $\calH$ only has to simulate at most $c^2$ binary relations.
\end{proof}

\begin{observation}\label{obs:subgraph}
 Let $\calG$ be a class of $S$-hypergraphs of bounded arity. Let $\calG'$ the closure of $\calG$ under edge deletions. Then there is a parsimonious slice reduction from $\sCQh[\calG']$ to $\sCQh[\calG]$.
\end{observation}
\begin{proof}
For every edge $e$ not appearing in a subhypergraph $\calH'$ but in $\calH$, let the respective relation on the right-hand-side contain all $k$-tuples of domain elements, where $k$ is the arity of $e$. Since $k$ is bounded by a constant, this can be done in polynomial time.
\end{proof}

The following lemma is an easy translation of a result from \cite{DurandM14} into our framework.

\begin{lemma}[\cite{DurandM14}]\label{lem:starsize}
 Let $\calG$ be a class of $S$-hypergraphs of unbounded $S$-star size. Then there is a counting slice reduction from $\sclique[\mathbb{N}]$ to $\sCQh[\calG]$.
\end{lemma}

\begin{proof}[Proof of Lemma~\ref{lem:strictSS}]
 For every $S$-hypergraph $(\calH, S)$ we compute an $S$-graph as follows: Take the primal graph of $\calH$ and delete all edges between vertices in $S$. Let $\calG'$ be the resulting class of $S$-graphs. Obviously, $\calG'$  has unbounded $S$-star size and thus there is a counting slice reduction from $\sclique[\mathbb{N}$ to $\sCQ[\calG']$. Using Observation \ref{obs:subgraph} and Observation \ref{obs:primal} then gives the desired result.
\end{proof}
}
\iftoggle{ARXIV}{\proofstrictSS}{}

\subsection{The main hardness results}

In this section we use the results of the last sections to prove the hardness results of Theorem~\ref{thm:trichotomy}.

The proof of Theorem~\ref{thm:bounded} in \cite{DurandM14} directly yields the following result.

\begin{lemma}\label{lem:DMreduction}
 Let $\calG$ be a class of $S$-hypergraphs of bounded arity. If the treewidth of $\calG$ is unbounded, then there is a counting slice reduction from to $\pn{Clique}[\mathbb{N}]$ to $\sCQh[\calG]$.
\end{lemma}

Combining Lemma~\ref{lem:DMreduction} and Corollary~\ref{cor:graphstostructures} yields the following corollary.
\begin{corollary}\label{cor:decision}
 Let $\calC$ be a class of queries such that the treewidth of the cores of the queries in $\calC$ is unbounded. Then there is a counting slice reduction from $\clique[\mathbb{N}]$ to $\sCQ[\calC]$.
\end{corollary}

\begin{lemma}\label{lem:contracthard}
 Let $\calG$ be a class of hypergraphs such that $\contract(\calG)$ is of unbounded treewidth. Then there is a counting slice reduction from $\sclique[\mathbb{N}]$ to $\sCQh[\calG]$.
\end{lemma}
\begin{proof}
 Assume first that $\calG$ is of unbounded strict $S$-star size. Then $\sCQh[\calG]$ is $\sW{1}$-hard by Lemma~\ref{lem:strictSS}. So we assume in the remainder of the proof that there is a constant $c$ such that for every $(\calH,S)$ in $\calG$ every $S$-component of $(\calH,S)$ contains only $c$ vertices from~$S$.
 
 From Theorem~\ref{thm:dalmauJ} it follows that there is a counting slice reduction from $\sclique[\mathbb{N}]$ to $\sCQh[\contract(\calG)]$. Therefore, it suffices to show parsimonious slice reduction $(U,r)$ from $\sCQh[\contract(\calG)]$ to $\sCQh[\calG]$ to show the lemma.
 
 The relation $U$ is defined as $U:=\{(\contract(\calH,S), (\calH,S))\mid (\calH,S) \in \calG\}$. By definition, this satisfies the covering condition.
 
For the definition of $r$, consider an instance $((\relA, S), \relB)$ of $\sCQh[\contract(\calG)]$ and let $\calH$ be the hypergraph of $\relA$. Moreover, let $(\calH',S)\in \calG$ be an $S$-hypergraph such that $(\calH,S) = \contract(\calH',S)$. W.l.o.g.~assume that for every edge $e$ of $\calH$, the structure $\relA$ contains one relation $R_e$ containing only a single tuple $\vec{e}$ where $\vec{e}$ contains the elements of $e$ in an arbitrary order. We construct an instance $r((\calH,S), (\calH',S), ((\relA, S), \relB)) = ((\relA', S), \relB')$. Similarly to $\relA$, the structure $\relA'$ has for every edge $e$ in $\calH'$ a relation $R_e$ that contains only a single tuple $\vec{e}$ with the properties as before. For every $S$-component $C$ of $\calH$ we do the following: Let $D_C$ be the tuples encoding the homomorphisms $h$ from $\relA[V(C)\cap S]$ to $\relB$. For every $v\in V(C)\setminus S$ we let $D_C$ be the domain of $v$. Whenever two elements $u,v\in V(C)\setminus S$ appear in an edge, we set $R^{\relB'}$ in such a way that for all tuples in $R_e^\relB$ the assignments to $u$ and $v$ coincide. Moreover, whenever $u\in V(C)\cup S$ and $v\in V(C)\setminus S$ we allow only tuples in which the assignment to $u$ coincides with the assignment to $u$ that is encoded in the assignment to $v$. For all edges $e$ of $\calH'$ with $e\setminus S\ne \emptyset$, we let $R_e^{\relB'}$ contain all tuples that satisfy the two conditions above. Finally, for all edges $e$ with $e\in S$ we set $R^{\relA'} :=R^\relA$.

It is easy to verify that $\hom(\relA, \relB, S) = \hom(\relA', \relB', S)$. Thus it only remains to show that the construction can be done in polynomial time. Note first that the number of variables from $S$ in any $S$-component of $\calH'$ is bounded by $c$. Thus we can compute all domains $D_C$ in time $\|\relB\|^{O(c)}$. The rest of the construction can then be easily done in polynomial time.
\end{proof}

\begin{corollary}\label{cor:countinghard}
 Let $\calC$ be a class of conjunctive queries such that $\contract(\calC)$ is of unbounded treewidth. Then there is a counting slice reduction from $\sclique[\mathbb{N}]$ to $\sCQ[\calC]$.
\end{corollary}
\begin{proof}
 This follows by combination of Lemma~\ref{lem:contracthard} and Corollary~\ref{cor:graphstostructures}.
\end{proof}

\subsection{Putting things together}\label{sct:puttingtogether}

We now finally show Theorem~\ref{thm:trichotomy} by putting together the results of the last sections.

\begin{proof}[Proof of Theorem~\ref{thm:trichotomy}]
1. follows directly from Corollary~\ref{cor:countingalg} with the observation that bounded treewidth of $\contract(\calC)$ implies bounded $S$-star size.

2. is Corollary~\ref{cor:decision} and Lemma~\ref{lem:decision} using the fact that $\CQ[\all]$ counting slice reduces to $\clique[\mathbb{N}]$; this follows from~\cite[Section 6.1]{FlumG06}.
 
 3. is Corollary~\ref{cor:countinghard}.
\end{proof}

\section{Conclusion}

In this paper we have proved a complete classification for the counting complexity of conjunctive queries, continuing a line of work that spans several previous papers~\cite{DalmauJ04,DurandM14,GrecoS14}. While this solves the bounded arity case completely, the most apparent open question is what happens for the unbounded arity case. This case is rather well understood for the decision version $\CQ$ of the problem~\cite{Marx13-tractablehypergraph}, but for counting not much is known. In particular, it is not known if the results of~\cite{Marx13-tractablehypergraph} can even be adapted to the quantifier free setting.

Another interesting problem would be to go from conjunctive queries to more expressive query languages. This has been done with some success for decision problems (see e.g.~\cite{Chen14} and the references therein), but the situation for counting is much less clear. It is known that counting and decision differ a lot at least in some settings, e.g.~even very simple unions of conjunctive queries yield hard counting problems~\cite{PichlerS13,thesis} while $\CQ$ for these queries is very easy. Can we get a better understanding of counting complexity in this setting and how it differs from decision?

To prove our results, we have extended the case complexity framework to counting complexity. We are very optimistic that this will be helpful when studying the research areas discussed above. Moreover, due to its generic nature, we feel that this framework should also be of use outside of the query answering context and allow transparent proofs and presentations in other areas of parameterized complexity.

\paragraph*{Acknowledgements}

Chen was supported by the Spanish project
TIN2013-46181-C2-2-R, 
by the Basque project GIU12/26,
and by the Basque grant UFI11/45.

\newpage

\bibliographystyle{plain}
\bibliography{CQ}

\newpage

\begin{appendix}

\section{Proofs of Section~\ref{sct:casecomplexity}}\label{app:case}

\begin{proof}[Proof of Theorem~\ref{thm:transitivity}]
Suppose that $(U_1, r_1)$ is a counting slice reduction from
$Q_1[S_1]$ to $Q_2[S_2]$, 
and that $(U_2, r_2)$ is a counting slice reduction from
$Q_2[S_2]$ to $Q_3[S_3]$.
We show that there exists a counting slice reduction
$(U, r)$ from $Q_1[S_1]$ to $Q_3[S_3]$.

Define $U \subseteq \Sigma^* \times \powfin(\Sigma^*)$
to be the set that contains a pair $(s_1, T_3)$
if and only if
there exists $T_2$ such that $(s_1, T_2) \in U_1$;
for each $t_2 \in T_2$, there exists
$T_{t_2}$ such that $(t_2, T_2) \in U_2$;
and,
it holds that $\bigcup_{t_2 \in T_2} T_{t_2} = T_3$.

We verify the coverage condition as follows.
For each $s_1 \in S_1$, there exists $T_2 \subseteq S_2$
such that $(s_1, T_2) \in U_1$ (since coverage holds for $U_1$),
and for each $t_2 \in T_2$,
there exists $T_{t_2}$ such that $(t_2, T_{t_2}) \in U_2$
(since coverage holds for $U_2$).
Hence, by definition of $U$, it holds that
$(s, \bigcup_{t_2 \in T_2} T_2) \in U$.

We verify the correctness condition as follows.
Let $A_1$ and $A_2$ by the algorithms given by the definition
of counting slice reduction for $r_1$ and $r_2$, respectively.
We describe an algorithm $A$ for the needed partial function $r$,
as follows; the algorithm $A$ uses the algorithms 
$A_1$ and $A_2$ in the natural fashion.
On an input $(s_1, T_3, x)$, the algorithm $A$
checks if $(s_1, T_3) \in U$; if so,
it may compute a set $T_2$ and sets $\{ T_{t_2} \}_{t_2 \in T_2}$
that witness this (in the sense of the definition of $U$).
The algorithm $A$ then invokes the algorithm $A_1$
on input $(s_1, T_2, x)$; each time $A_1$
makes an oracle query,
it is of the form $Q_2(t_2, z)$ where $t_2 \in T_2$,
and to resolve the query, the algorithm $A$ calls
$A_2$ on input $(t_2, T_{t_2}, z)$.

The time analysis of the algorithm $A$ is analogous to
that carried out in~\cite[Appendix Section A]{Chen14}.
\end{proof}

 \begin{proof}[Proof of Theorem~\ref{thm:slice-red-gives-fpt-red}]
Let $(U, r)$ be the counting slice reduction.
Since $S$ and $S'$ are both computable,
there exists a computable function $h: \Sigma^* \to \powfin(\Sigma^*)$
such that for each $s \in S$, it holds that $(s, h(s)) \in U$.
Consider the algorithm $A'$ that does the following:
given $(s, y)$, check if $s \in S$; if not, then output $0$,
else compute $r(s, h(s), y)$ using the algorithm $A$ for $r$
guaranteed by the definition of counting slice reduction.
This algorithm $A'$ and $h$ give a counting FPT-reduction
from $\param{Q[S]}$ to $\param{Q'[S']}$.
Note that, as a function of $(s,y)$,
we have that $A'(s,y)$ is FPT-computable with respect to $\pi_1$,
since $r(s,h(s),y)$ is FPT-computable with respect to $(\pi_1,
\pi_2)$,
and $h(s)$ is a function of $s$.
\end{proof}

\begin{proof}[Proof of Theorem~\ref{thm:promfpt-to-fpt}]
Let $(U, r)$ be the counting slice reduction given by hypothesis.
Since $U$ and $S$ are both computably enumerable,
there exists
 a computable function $h: \Sigma^* \to \powfin(\Sigma^*)$
such that for each $s \in \Sigma^*$, it holds that $(s, h(s)) \in U$.
Consider the following algorithm $A'$ for $K$:
on an input $(s, y)$,
 compute $r(s, h(s), y)$ using the algorithm $A$ for $r$
guaranteed by the definition of counting slice reduction;
the oracle queries that the algorithm $A$ poses to $Q$ are resolved
by the algorithm $B$ which witnesses $\paramprom{Q[S]} \in \prom{\FPT}$.
Suppose that the running time of $B$ on $(t, z)$ is bounded above by
$f(t)p(|(t,z)|)$; note that the time needed to resolve an oracle query
$(t,z)$
made by $A$ on $(s,y)$ is
$(\max_{t \in h(s)} f(t)) p(|(t,z)|)$, and that 
$\max_{t \in h(s)} f(t)$ is a computable function of $s$.
We then have that $A'(s, y)$ is FPT-computable with respect to $\pi_1$,
by an argument similar to that in the proof of
Theorem~\ref{thm:slice-red-gives-fpt-red}.
\end{proof}

\iftoggle{ARXIV}{}{
\section{Proofs of Section~\ref{sct:algorithms}}
\subsection{Proof of Lemma~\ref{lem:computecores}}\label{app:computecores}

\begin{proof}[Proof of Lemma~\ref{lem:computecores}]
\prooflemcomputecores
}

\iftoggle{ARXIV}{}{
\subsection{Proof of Lemma~\ref{lem:decision}}\label{app:decision}

\prooflemdecision
}


\iftoggle{ARXIV}{}{
\section{Proof of Lemma~\ref{lem:strictSS}}

\proofstrictSS
}

\end{appendix}

\end{document}